\RequirePackage{fix-cm}

\documentclass[smallcondensed
]{svjour3}     
\smartqed  
\usepackage{graphicx}
\usepackage{epstopdf}
\usepackage{geometry}
\usepackage{comment}
\usepackage{xspace}
\usepackage{amsmath}
\usepackage{amssymb}

\usepackage{epsfig}
\usepackage{subfigure}
\usepackage{algorithmic}
\usepackage{graphics}
\usepackage{alltt}
\usepackage{float}
\usepackage{array}
\usepackage{footmisc}
\usepackage{comment}
\usepackage{algorithm}

\usepackage[justification=centering]{caption}
\usepackage[misc]{ifsym}

\newcommand{\minwise}{MinWise\xspace}
\newcommand{\oph}{OPH\xspace}
\newcommand{\goph}{GOPH\xspace}
\newcommand{\hoph}{HOPH\xspace}
\newcommand{\abhoph}{(a:b)HOPH\xspace}
\newcommand{\oohoph}{(1:1)HOPH\xspace}
\newcommand{\othoph}{(1:2)HOPH\xspace}
\newcommand{\tohoph}{(2:1)HOPH\xspace}

\begin{document}

\title{Hierarchical One Permutation Hashing: Efficient Multimedia Near Duplicate Detection
}


\author{Chengyuan Zhang$^\dagger$ \and Yunwu Lin$^\dagger$ \and Lei Zhu$^\dagger$ \and XinPan Yuan$^\ddagger$ \and Jun Long$^\dagger$         \and Fang Huang$^\dagger$
}


\institute{Chengyuan Zhang \at
              \email{cyzhang@csu.edu.cn}           
           \and
           Yunwu Lin \at
              \email{lywcsu@csu.edu.cn}
           \and
           Lei Zhu \at
              \email{leizhu@csu.edu.cn}           
           \and
           \Letter XinPan Yuan \at
              \email{xpyuan@hut.edu.cn}
              \and
           \Letter Jun Long \at
              \email{jlong@csu.edu.cn}           
           \and
           Fang Huang \at
              \email{hfang@csu.edu.cn}
              \and
           $^\dagger$ School of Information Science, Central South University, PR China\\
$^{\ddagger}$ School of Computer, Hunan University of Technology, China\\
}

\date{Received: date / Accepted: date}

\maketitle

\begin{abstract}
With advances in multimedia technologies and the proliferation of smart phone, digital cameras, storage devices, there are a rapidly growing massive amount of multimedia data collected in many applications such as multimedia retrieval and management system, in which the data element is composed of text,
image, video and audio. Consequently, the study of multimedia near duplicate detection has attracted significant concern from research organizations and commercial communities. Traditional solution minwish hashing (\minwise) faces two challenges: expensive preprocessing time and lower comparison speed. Thus, this work first introduce a hashing method called one permutation hashing (\oph) to shun the costly preprocessing time. Based on \oph, a more efficient strategy group based one permutation hashing (\goph) is developed to deal with the high comparison time. Based on the fact that the similarity of most multimedia data is not very high, this work design an new hashing method namely hierarchical one permutation hashing (\hoph) to further improve the performance. Comprehensive experiments on real multimedia datasets clearly show that with similar accuracy \hoph is five to seven times faster than \minwise .

\keywords{Hierarchical one permutation hashing \and Multimedia data \and Near duplicate detection}

\end{abstract}

\section{Introduction}
\label{intro}
Due to the proliferation of smart phone, video cameras, storage devices as well as the development of social networks, large volumes of multimedia data associating to text, image, video and audio have been generated by users every day. For instance, it is reported that there are over 95 million geo-tagged photos on Flickr with a daily growth rate of around 500,000 new geo-tagged photos. As a result, in recent years variety of multimedia information retrieval applications have emerged and become more and more popular.

As a vital component of multimedia information retrieval applications, near duplicate detection, also known as similarity search, is a well established research topic and still attracts lots of attention. An element is called a near-duplicate of a reference element if it is ¡°close¡±, according to some defined measure, to the reference element. For example, given an image and a collection, we want to find those candidates in the collections which are most similar to the input image based on their Jaccard similarity.

As efficient solution of multimedia near duplicate detection, \minwise approach has been paid substantial effort in recent research literatures. For example, in text near duplicate detection, $k$ independent random permutations are generated to convert the words of each document into a set of fingerprints to accelerate similarity search; similarly, in image near duplicate detection, $k$ independent random permutations are applied to convert the visual words of each image into a set of fingerprints.

Though \minwise method are successful in large scale similar multimedia searching, it still faced two critical challenges. Firstly, the preprocessing cost of \minwise can be very expensive. To ensure the accuracy of \minwise, hundreds of independent random permutations are generated. Secondly, \minwise is time consuming, because it has to compare all fingerprints to calculate the similarity.

In order to reduce preprocessing time of \minwise, we first introduce a hashing method \oph, which is widely used in text search area, to multimedia near duplicate detection. Different with \minwise, which requires hundreds of random permutation, \oph only require single permutation, which significantly reduces the time consumption during random permutation. Based on \oph, we further divide the fingerprints into several groups, and bring in the concept of small probability event. With the assistance of small probability event, our proposed G\oph can terminal the comparison earlier, but with similar accuracy. Meanwhile, we observe that most of the data are not similar to the input in real application. Thus, based on this observation, we purpose a new hashing method namely \hoph to further reduce the comparison time.

\textbf{Contributions.} The principle contributions of this paper are summarized as follows.
\begin{itemize}
\item \oph is introduced to avoid the expensive preprocessing step of \minwise. Based on \oph, an efficient group based algorithm called \goph is develop to reduce the comparison time.
\item To further accelerate the search speed, we proposed a new hashing method \hoph.
\item Comprehensive experiments on real and synthetic text and image datasets demonstrate that our proposed hashing method \hoph achieve substantial improvements over \minwise, \oph, but with similar accuracy.
\end{itemize}

\textbf{Roadmap.} The rest of the paper is organized as follows. Section ~\ref{relwork} formally defines the problem of multimedia near duplicate detection, followed by the introduction of the related work. Section ~\ref{sec:one hashing} describes \oph, \goph and \goph's image application. Section ~\ref{sec:hierarchical one hash} presents \hoph, some important theoretical analysis of \hoph and its image application. Extensive experiments are depicted in Section ~\ref{perform}. Finally, Section ~\ref{con} concludes the paper.

\section{Related Work}
\label{relwork}

In this subsection, we present some existing techniques, such as \minwise, $b$-bit \minwise, for the problem of text near-duplicate detection. Then we also present some general techniques related to image similarity computation.

\subsection{Near-duplicate detection for text}
Near-duplicate detection is one of the key problems in the area of database and information retrieval. It aims to detect groups of documents with almost the same contents among a document collection. Two documents with a great amount of shared attributes do not necessarily count as near-duplicate. For the near-duplicate detection problem, Yang et al.\cite{DBLP:conf/sigir/YangC06} proposed an instance-level constrained clustering approach. Their framework incorporates information such as document attributes and content structure into the clustering process to form near-duplicate clusters. Yang et al.~\cite{DBLP:journals/jasis/HoadZ03,YangINF2013} emphasized that the gap between rare words¡¯ term frequency in two documents should be smaller than that between common words¡¯ and their best ranking is giving by a term weighting function biased towards rare terms. Hassanian-esfahani et al.~\cite{DBLP:journals/eswa/Hassanian-esfahani18} proposed a Sectional \minwise (S-\minwise) for the detection of near-duplicate documents to enhances the \minwise data structure with information about the location of the attributes in the document.

\minwise~\cite{DBLP:journals/jcss/BroderCFM00,LINYANG13,DBLP:journals/siamdm/Vsemirnov04} is a locality sensitive hashing for the Jaccard similarity, which is a most popular technique for efficiently text similarity computing. It has a wide range of applications such as duplicate detection~\cite{DBLP:conf/sigir/Henzinger06}, nearest neighbor search~\cite{DBLP:conf/stoc/IndykM98}, large-scale learning~\cite{DBLP:conf/nips/LiSMK11}, all-pairs similarity~\cite{DBLP:conf/www/BayardoMS07}, etc.. Border et al.~\cite{DBLP:journals/cn/BroderGMZ97} invented the \minwise algorithm for near-duplicate web page detection and clustering. Li et al.~\cite{DBLP:conf/nips/LiSMK11} presented a simple effective solution to large-scale learning in massive and extremely high-dimensional datasets and indicated that $b$-bit \minwise is significantly more accurate than Vowpal Wabbit in binary data. The major defect of \minwise algorithm and $b$-bit \minwise is that they require an expensive preprocessing step, by conducting $k$ permutations on the entire dataset~\cite{Li2012One}. Pagh et al.~\cite{DBLP:conf/pods/PaghSW14} studied and addressed the question: How many bits is it necessary to allocate to each summary in order to get an estimate with $1 ¡À ¦Å$ relative error.

\subsection{Similarity computation for image}
Similarity computation~\cite{DBLP:conf/cikm/WangLZ13,DBLP:conf/mm/WangLWZZ14} for image is another important technique which is focused by a great many of researchers. The first significant problem concerning efficient similarity measure is image~\cite{DBLP:conf/ijcai/WangZWLFP16} representations. Scale Invariant Feature Transform (SIFT for short)~\cite{DBLP:journals/ijcv/Lowe04} proposed by Lowe is a classical approach in image recognition and computer vision area. It aims to detects and describes local features in images~\cite{DBLP:conf/iccv/Lowe99}. The SIFT descriptor is invariant to uniform scaling, orientation, illumination changes and partially invariant to affine distortion. Many studies of image processing~\cite{DBLP:conf/pakdd/WangLZW14,NNLS2018} and retrieval~\cite{DBLP:conf/mm/WangLWZ15} use it as one of the basic methods. In order to identify and remove the most false positive matches, Zhou et al.~\cite{DBLP:conf/icimcs/ZhouLWLT12} proposed to generate binary SIFT descriptor in a given pair of the images from the original SIFT descriptor. Zhou et al.~\cite{DBLP:conf/icip/ZhouLZ16} extended SIFT-based match kernels by integrating the match functions for SIFT and CNN features~\cite{DBLP:journals/cviu/WuWGHL18,DBLP:journals/pr/WuWGL18,TC2018,DBLP:journals/pr/WuWLG18}. In order to improve the performance of image retrieval, they proposed a threshold exponential match kernel for CNN features to filter out the images whose the semantic similarity is lower than the threshold. To obtain enhanced performance, Zhang et al. ~\cite{DBLP:journals/pami/ZhangYWLT15} utilized 1000 semantic attributes to revise the vocabulary tree of SIFT descriptors in Bag-of-Word model (BoW for short), which stores only occurrence counts of vector quantized features~\cite{DBLP:conf/icimcs/QuSYL13}. The Bag-of-Visual-Word model (BoVW for short) represent a image by a set of visual words applying SIFT descriptor and K-means clustering algorithm~\cite{DBLP:conf/cvpr/NisterS06,DBLP:conf/cvpr/PhilbinCISZ07}. BoVW with $tf-idf$ weighting~\cite{DBLP:conf/iccv/SivicZ03} has proven to be a very successful approach for image and particular object retrieval. Nister et al. proposed a recognition scheme by applying this modal, which scales efficiently to a large number of objects. There are two types of searching methods based on BoVW modal, namely the exact inverted file methods~\cite{DBLP:journals/csur/ZobelM06} and the hashing methods~\cite{DBLP:conf/civr/ChumPIZ07}. Hashing methods~\cite{DBLP:conf/sigir/WangLWZZ15,DBLP:journals/ivc/WuW17} are used to solve the problem of image similarity measure~\cite{DBLP:journals/corr/abs-1708-02288} and multimedia retrieval~\cite{DBLP:journals/tip/WangLWZ17,DBLP:journals/tip/WangLWZZH15,DBLP:journals/tnn/WangZWLZ17}, which are concerned by the community. The spectral hashing method proposed by Shao et al.~\cite{DBLP:journals/prl/ShaoWOZ12} and the local sensitive hashing method designed by are belong to the unsupervised hashing methods. Jain et al.~\cite{DBLP:conf/cvpr/JainKG08} a method that applied a Mahalanobis distance function that captures the imagespsila underlying relationships well. This approach combined the \minwise method with distance metric learning. Li et al.~\cite{DBLP:journals/tmm/LiWCXL13} presented a method to directly optimize the graph Laplacian by using spectral hashing combined with a distance learning.

For problem of detection of near duplicate images, Chum et al.~\cite{DBLP:conf/bmvc/ChumPZ08} proposed an efficient way based on a \minwise method, which uses a visual vocabulary of vector quantized local feature descriptors and for retrieval exploits enhanced \minwise techniques. Zhang et al.~\cite{DBLP:conf/mm/ZhangC04} applied a parts-based representation of each scene by building Attributed Relational Graphs (ARG) between interest points. Based on Stochastic Attributed Relational Graph Matching, they compared the similarity of two images. Torralba et al.~\cite{DBLP:conf/cvpr/TorralbaFW08} proposed a method to learn short descriptors to retrieve similar images from a huge database, which is based on a dense 128D global image descriptor. Jain et al.~\cite{DBLP:conf/cvpr/JainKG08} introduced a method for efficient extension of Locally Sensitive Hashing scheme for Mahalanobis distance. Apparently, both above approaches use bit strings as a fingerprint of the image. Philbin et al.~\cite{DBLP:conf/cvpr/PhilbinCISZ07} proposed a large-scale object retrieval system and compared different scalable methods for building a vocabulary. Besides, they introduced a novel quantization method based on randomized trees to enhance the performance of image-feature vocabularies construction.
\section{Basic Apporach}
\label{sec:one hashing}
In this section, we first formally define the problem of
multimedia near duplicate detection, and then review the hashing method \oph, which is widely used in text search area, in section ~\ref{sec:one hashing review}. We propose an advanced hashing method called \goph to improve the performance of \oph in section ~\ref{sec:GOPH}. Section ~\ref{sec:image detection} presents the image application of \goph. Table~\ref{tab:nnd notation} below summarizes the mathematical notations
used throughout this paper.

\begin{table}
	\centering
    \small
	\begin{tabular}{|p{0.27\columnwidth}| p{0.62\columnwidth} |}
		\hline
		\textbf{Notation} & \textbf{Definition} \\ \hline\hline
		~$q$             & a multimedia data (query)                                \\ \hline
        ~$\mathcal{D}$   & the multimedia dataset                                \\ \hline
	  	~$\mathcal{T}$   & the similarity threshold                                \\ \hline
	  	~$\psi$   & a random permutation                                \\ \hline
	  	~$\mathcal{V}$   & the whole vocabulary space                                \\ \hline
	  	~$\mathcal{N}_mb$   & the number of \textquotedblleft matched bins\textquotedblright                                \\ \hline
	  	~$\mathcal{N}_eb$   & the number of \textquotedblleft empty bins\textquotedblright                               \\ \hline
	  	$\hat{\mathcal{R}}_{mb}$   & unbiased estimator                                \\ \hline
	  	$var(\hat{\mathcal{R}}_{mb})$   & variance of \oph                               \\ \hline
	  	$\mathcal{K}$   &  the bin number of the comparison part                               \\ \hline
	  	$\mathcal{X}$   &  the number of times that the fingerprints are equal     \\ \hline
        $T$   &  the number of match bin after $\mathcal{K}$ comparisons     \\ \hline
        $\epsilon$  &  the error tolerance      \\ \hline
        $\mathcal{I}$   & the image dataset                                \\ \hline
        $\mathcal{N}_{mb_{h}}$ & The number of matched bin of \hoph                                \\ \hline
        $\mathcal{N}_{eb_{h}}$ & The number of empty bin of \hoph                                \\ \hline
        $\hat{\mathcal{R}}_{mb_{h}}$ & The unbiased estimator  of \hoph                                \\ \hline
        $\otimes$ & the empty bin of \oph and \hoph                                \\ \hline
	\end{tabular}
    \vspace{2mm}
    \caption{Notations} \label{tab:nnd notation}	
    \vspace{-6mm}
\end{table}

\subsection{Problem Definition}
\label{prob def}
\textbf{Near Duplicate.} Given a multimedia data $q \in \mathcal{D}$, any multimedia data $p \in \mathcal{D}$ such that $sim(q, p) < \mathcal{T}$ is a near duplicate
of $q$, where $\mathcal{T}$ is a similarity threshold. Among the available
similarity comparison function \textbf{sim} which might be exploited, the Jaccard similarity
has been chosed in this work, since it has been widely used in different applications.
Without loss of generality, this paper mainly consider two types multimedia data, text and
image.

\subsection{One Permutation Hashing Review}
\label{sec:one hashing review}
To reduce the times of random permutation of \minwise, Ping Li \emph{el al.} propose an signature named one permutation hashing  in ~\cite{Li2012One}. To generate hundreds of samples, traditional signature \minwise requires hundreds of random permutation. However, only single permutation is enough for \oph, which significantly reduces the time consumption during random permutation.

\textbf{One Permutation Hashing.} First, a random permutation $\psi$ is generated.
For each document $\mathcal{D}_i$ a one permutation hashing min $\psi(\mathcal{D}_i)$ is recorded. Consider $\mathcal{D}_1$, $\mathcal{D}_2$, $\mathcal{D}_3$ $\subseteq$ $\mathcal{V}$ = \{0, 1, ..., 16\}. Assume $\mathcal{D}_1$ = \{1, 2, 5, 10, 12, 15\}, $\mathcal{D}_2$ = \{1, 2, 6, 10, 12, 14\}, $\mathcal{D}_3$ = \{2, 9, 10, 12, 14\}, we apply the permutation $\psi$ on the three sets and present the corresponding $\psi(\mathcal{D}_1)$, $\psi(\mathcal{D}_2)$, $\psi(\mathcal{D}_3)$ as binary (0/1) vector as what is shown in Fig. \ref{fig:oph_example}.


\begin{figure}[thb]
\newskip\subfigtoppskip \subfigtopskip = -0.1cm
\centering
\includegraphics[width=.90\linewidth]{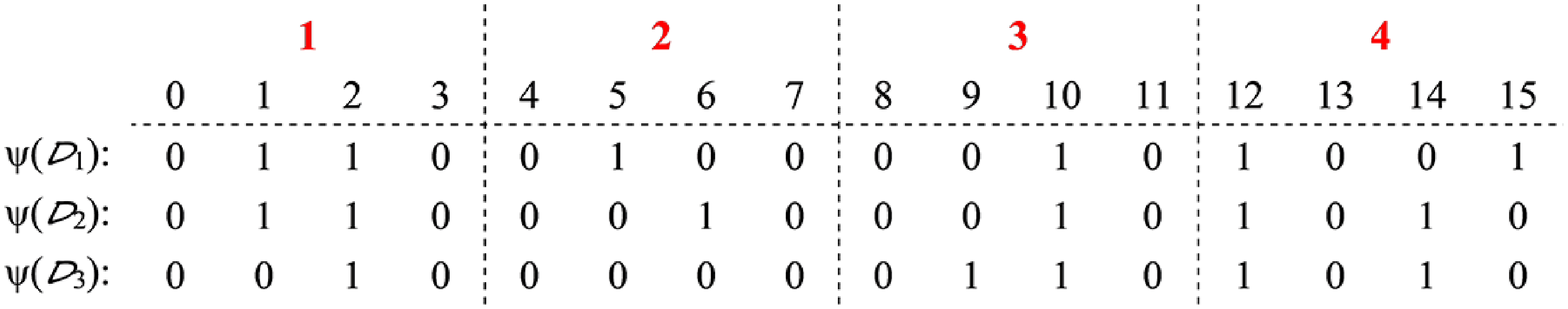}
\vspace{-1mm}
\caption{\small One Permutation Hashing Example }
\label{fig:oph_example}
\end{figure}



Then, we divide the whole space $\mathcal{V}$ into $k$ bins, and select the first non-zero element as a sample for each bin. In some special case, if there is no non-zero element in the bin, $\otimes$ is used to represent the empty bin. By this way, a random permutation of \oph can generate $k$ sample through $k$ bins. Assume $k$=4, the sample selected from  $\psi(\mathcal{D}_1)$ is [1, 5, 10, 12], $\psi(\mathcal{D}_2)$ is [1, 6, 10, 12], and $\psi(\mathcal{D}_3)$ is [2, $\otimes$, 9, 12]. Because we only need to compare the sample within the same bin, we can use the smallest possible representation to represent the actual value. For example, for $\psi(\mathcal{D}_1)$, the final representation is [1, 1, 2, 0]; for $\psi(\mathcal{D}_2)$, the final representation is [1, 2, 2, 0]; similarly, for $\psi(\mathcal{D}_3)$, the final representation is [2, $\otimes$, 1, 0].

From the above example, the sets $\psi(\mathcal{D}_1)$ and $\psi(\mathcal{D}_2)$ have 3 identical smallest possible representations and the estimated similarity will be 0.75, while the exact similarity is 0.5. The sets $\psi(\mathcal{D}_1)$ and $\psi(\mathcal{D}_3)$ share one smallest possible representation and their similarity estimate is 0.333 (0.375 is exact).

Finally, we introduce some important properties of \oph. Without loss of generality, we consider two documents $\mathcal{D}_1$ and $\mathcal{D}_2$. Firstly, it is the two fundamental definition of \oph $\mathcal{N}_{mb}$ and $\mathcal{N}_{eb}$. $\mathcal{N}_{mb}$ and and $\mathcal{N}_{eb}$ represents the number of \textquotedblleft matched bins\textquotedblright and the number of \textquotedblleft empty bins\textquotedblright respectively:

\begin{equation}
\mathcal{N}_{mb}=\sum_{j=1}^{k}\mathcal{B}_{mb,i}
\end{equation}

\begin{equation}
\mathcal{N}_{eb}=\sum_{j=1}^{k}\mathcal{B}_{eb,i}
\end{equation}

where $\mathcal{B}_{mb,i}$ and $\mathcal{B}_{eb,i}$ are defined for the $i$-th bin, as

\begin{equation}
\mathcal{B}_{mb,i}=\begin{cases}
            1 \ if \ min(\psi(\mathcal{D}_1))=min(\psi(\mathcal{D}_2)) \& \psi(\mathcal{D}_1) \neq \otimes \& \psi(\mathcal{D}_1) \neq \otimes \ in \\ \ \ the \ i-th \ bin\\
            0 \ otherwise
           \end{cases}
\end{equation}

\begin{equation}
\mathcal{B}_{eb,i}=\begin{cases}
            1 \ if \ \psi(\mathcal{D}_1)=\psi(\mathcal{D}_2) = \otimes \ in \ the \ i-th \ bin\\
            0 \ otherwise
           \end{cases}
\end{equation}

Denote $\psi$ a random permutation: $\psi :$ $\mathcal{V}$ $\to$ $\mathcal{V}$. The hashed values are the two smallest possible representation sets after applying the permutation $\psi$ on $\psi(\mathcal{D}_1)$ and $\psi(\mathcal{D}_2)$. The probability at which the two hashed values are equal is

\begin{equation}
\mathcal{R}=\textbf{Pr}(min(\psi(\mathcal{D}_1))= min(\psi(\mathcal{D}_2))) = \dfrac{\mathcal{D}_1 \cap \mathcal{D}_2}{\mathcal{D}_1 \cup \mathcal{D}_2}=Jaccard(\mathcal{D}_1,\mathcal{D}_2)
\end{equation}

Then the unbiased estimator of one permutation hashing is

\begin{equation}
\label{eqn:oph est}
\hat{\mathcal{R}}_{mb}=\dfrac{\mathcal{N}_{mb}}{k-\mathcal{N}_{eb}}\\
\end{equation}

\begin{equation}
E(\hat{\mathcal{R}}_{mb})=\mathcal{R}
\end{equation}
Based on unbiased estimator, assume $g$ = $|\mathcal{D}_1 \cup \mathcal{D}_2|$ the variance of \oph is
\begin{equation}
var(\hat{\mathcal{R}}_{mb})=\mathcal{R}(1-\mathcal{R})(E(\dfrac{1}{k-\mathcal{N}_{eb}})(1+\dfrac{1}{g-1})-\dfrac{1}{g-1})
\end{equation}

\subsection{Group Based One Permutation Hashing}
\label{sec:GOPH}
\textbf{Motivation.} While calculating the similarity, we have to compare the fingerprints one by one in one permutation hashing, which is time consuming. Meanwhile, one permutation hashing is claimed to satisfy the binomial distribution in ~\cite{Li2012One}. Thus, based on this property, we design a new algorithm named group based one permutation hashing(\goph) to reduce the comparison time.

\textbf{Basic Idea.} After applying one permutation to documents, we first aggregate the generated fingerprints into $n$-groups. Then, from the first group to the last group, we progressively compute the similarity between the corresponding groups, and bring in the concept of small probability event as a filter to accelerate the comparison for the remain groups.




Assume $\mathcal{K}$ is the bin number of the comparison part, $\mathcal{X}$ is defined as the number of times that the fingerprints are equal in comparison part,  denotes as
\begin{equation}
\mathcal{X}=\sum_{j=0}^{\mathcal{K}}1\{min(\psi(\mathcal{D}_i))= min(\psi(\mathcal{D}_i))\}
\end{equation}

Assume $T$ is the estimator of the number of match bin after $\mathcal{K}$ comparisons.
Because there are only two situations: "match" or "unmatch" in each comparison, these two situations are opposite each other and independent from each other, and the comparison result is not related to the results of other comparisons.
Obviously, $\mathcal{X}$ satisfies the binomial distribution, denotes as $\mathcal{X} \sim F($T$,\mathcal{K})$. Thus, The distribution function $\mathcal{F}(x)$ of the variable $\mathcal{X}$ is denoted as follow:


\begin{equation}
\mathcal{F}(x)=\begin{cases}
            \sum_{j=0}^{\mathcal{X}}\tbinom{\mathcal{K}}{j}{T}^{j}(1-{T})^{\mathcal{K}-j} \ \ \ \ ,  x < \mathcal{X}\\
            \sum_{j=\mathcal{X}+1}^{\mathcal{K}}\tbinom{\mathcal{K}}{j}{T}^{j}(1-{T})^{\mathcal{K}-j}, x \geq \mathcal{X}
           \end{cases}
\end{equation}

In the following, we first introduce the concept of small probability event, then introduce how the small probability event is used to avoid unnecessary comparison for the remain groups.

\begin{definition}[\textbf{Small Probability Event}] \label{def:low bound}
Given an error tolerance $\epsilon$, $\mathcal{F}(x)$ is the distribution function of variable $\mathcal{X}$, and $\mathcal{X}\sim F( T ,\mathcal{K})$, an event is called a small probability event, if and only if $\mathcal{F}(x) \leq \epsilon$.
\end{definition}

Assume $\mathcal{M}_a$ indicates the excepted average match bins for each groups, and $\mathcal{M}_{r_a}$ represents the minimum average of the remainder groups, if the final  average match bins not less than $\mathcal{M}_a$ according to the current binomial distribution. Assume the event $E$ indicates the probability that $\mathcal{M}_{r_a}$ can eventually reach $\mathcal{M}_a$, if the probability of event $E$ is less than the error tolerance $\epsilon$, we can see event $E$ as a small probability event, and terminate the subsequent calculation after corresponding processing. More specifically, while comparing the value of $\mathcal{M}_a$ and $\mathcal{M}_{r_a}$, we have the following two cases:

i) If $\mathcal{M}_{r_a} < \mathcal{M}_a$, for the remain groups, we compute whether $\mathcal{F}(\mathcal{M}_{r_a})$ is smaller than $\epsilon$. If $\mathcal{F}(\mathcal{M}_{r_a})$ is smaller than $\epsilon$,  we consider that it is almost impossible for these two documents to meet the similarity requirement, label them as dissimilarity documents, and terminal the algorithm; otherwise, we takes next group into further consideration.

ii) If $\mathcal{M}_{r_a} \geq \mathcal{M}_a$, for the remain groups, we still compare the value of $\mathcal{F}(\mathcal{M}_{r_a})$ and  $\epsilon$. If $\mathcal{F}(\mathcal{M}_{r_a})$ is smaller than $\epsilon$,  we consider that these two documents should be a similar document pair, add them to result set, and  terminal the algorithm; otherwise, we takes next group into further consideration.

\textbf{Algorithm.} Algorithm \ref{alg:dpc} illustrates the implementation details of the
document pair comparison based on \goph. In Line \ref{alg:dpc init},we first initialize current group number $n_c$ to 1, current match bins $\mathcal{M}_c$ to 0, and calculate $\mathcal{M}_a$ by input value $n$, $k'$, $\mathcal{T}$. Then, for each group of document pair $(\mathcal{D}_i, \mathcal{D}_j)$, we gradually calculate their similarity. From Line
\ref{alg:cal cur info st} to Line \ref{alg:cal cur info end}, we update the $\mathcal{M}_c$ based on current group's information. After computing the value of $\mathcal{M}_{r_a}$ in Line \ref{alg:cal mra}, we compare its value with $\mathcal{M}_a$ . From Line \ref{alg:comp small st} to Line \ref{alg:comp small end}, if $\mathcal{M}_{r_a} < \mathcal{M}_a$, we further compare $F(\mathcal{M}_{r_a})$ with $\epsilon$, if its value is larger than $\epsilon$, we continue to compare the similarity of next group; Otherwise, we consider current document pair is a similar pair, and break the algorithm. From Line \ref{alg:comp large st} to Line \ref{alg:comp large end}, If $\mathcal{M}_{r_a} \geq \mathcal{M}_a$, and $F(\mathcal{M}_{r_a})$ is larger than $\epsilon$, we continue to compare the similarity of next group; Otherwise, we break the algorithm and consider current document pair is not similar. Following is an example of \goph comparison algorithm.

\begin{algorithm}
\begin{algorithmic}[1]
\footnotesize
\caption{\bf \goph Comparison}
\label{alg:dpc}

\INPUT  \\
$(\mathcal{D}_i, \mathcal{D}_j)$: the document pair;
$\mathcal{T}$: user preferred similarity threshold, $k'$: bin number of each group; \\
$n$: group number of the document;
$\epsilon$: largest error tolerance;
A small probability $e$; \\

\OUTPUT \\
$\mathcal{O}$: Similar document pairs\\

\STATE $l$=1; $\mathcal{M}_c$=0; $\mathcal{M}_a$=$k'\mathcal{T}$;
\label{alg:dpc init}
\FOR{each group of document pair $(\mathcal{D}_i, \mathcal{D}_j)$}
    \FOR{each bins in $l$-th group}
    \label{alg:cal cur info st}
        \IF{the bin has equal value}
            \STATE $\mathcal{M}_c$++;
        \ENDIF
    \ENDFOR
    \STATE $l$++;
    \label{alg:cal cur info end}
    \STATE $\mathcal{M}_{r_a}=\dfrac{nk'\mathcal{T}-\mathcal{M}_c}{n-l}$;
    \label{alg:cal mra}
    \IF{$\mathcal{M}_{r_a} < \mathcal{M}_a$}
    \label{alg:comp small st}
    \IF{$F(\mathcal{M}_{r_a})>\epsilon$}
            \STATE $Continue$;
        \ELSE
            \STATE $\mathcal{O} \leftarrow (\mathcal{D}_i, \mathcal{D}_j)$ ;
            \STATE $Break$;
        \ENDIF
    \ENDIF
    \label{alg:comp small end}
    \IF{$\mathcal{M}_{r_a} \geq \mathcal{M}_a$}
    \label{alg:comp large st}
        \IF{$F(\mathcal{M}_{r_a})>\epsilon$}
            \STATE $Continue$;
        \ELSE
            \STATE $Break$;
        \ENDIF
    \ENDIF
    \label{alg:comp large end}
\ENDFOR

\end{algorithmic}
\end{algorithm}

\begin{example}
\label{ex:motivation}
Given a document pair $(\mathcal{D}_i, \mathcal{D}_j)$, assume $k'$=100, $n$=10, $\mathcal{T}$=0.6, $\epsilon={10}^{-4}$ and $\mathcal{X} \sim F(\mathcal{T},\mathcal{K})$. Then, the values of $F(\mathcal{X})$ for different $\mathcal{X}$ are shown in Table \ref{tab:prob} and the value of $\mathcal{M}_a$ is 60. Assume there are 65 bins matching in the 1-st
comparison, the value of $\mathcal{M}_{r_a}$ is 59.4. Because  $\mathcal{M}_{r_a} < \mathcal{M}_a$, and $F(\mathcal{M}_{r_a})$ is larger than $\epsilon$, we continue to compare next group. Assume only 5 bins matching in the 2-nd
comparison, the value of $\mathcal{M}_{r_a}$ is 66.25. Because  $\mathcal{M}_{r_a} \geq \mathcal{M}_a$, and $F(\mathcal{M}_{r_a})$ is larger than $\epsilon$, we continue to compare next group. Assume in the 3-rd
comparison, there are 10 bins matching, similarly, the value of $\mathcal{M}_{r_a}$ is 74.28. Because  $\mathcal{M}_{r_a} \geq \mathcal{M}_a$, and $F(\mathcal{M}_{r_a})$ is larger than $\epsilon$, we continue to compare next group. In the 4-th comparison, assume there are 10 bins matching, the value of $\mathcal{M}_{r_a}$ is 85. Obviously, $\mathcal{M}_{r_a} \geq \mathcal{M}_a$, we continue to compare the value of $F(x\geq85)$ and $\epsilon$. Because the probability of $F(x\geq85)$ equals 5.0732E-08, which  is smaller than $\epsilon$, it means the probability that the document pair $(\mathcal{D}_1, \mathcal{D}_2)$ is a similar pair is a small probability event. Thus, we terminal the algorithm and consider current document pair is not similar pair.

\end{example}

\begin{table}
\centering
\caption{Probability distribution of ${x\leq \mathcal{X}}$ and ${x > \mathcal{X}}$ when $k'$=100, $n$=10, $\mathcal{T}$=0.6} \label{tab:prob}

\centering
\begin{tabular}{ccc|ccc}

\hline
$\mathcal{X}$& $F(x < \mathcal{X})$& $F(x \geq \mathcal{X})$ &$\mathcal{X}$& $F(x < \mathcal{X})$& $F(x \geq \mathcal{X})$\\
\hline
5&	3.27948E-33&	1&	55&	0.131090453&	0.868909547\\
10&	1.25639E-26&	1&	60&	0.456705514&	0.543294486\\
15&	2.31928E-21&	1&	65&	0.820530647&	0.179469353\\
20&	5.56419E-17&	1&	70&	0.975217177&	0.024782823\\
25&	2.71442E-13&	1&	75&	0.998810999&	0.001189001\\
30&	3.46422E-10&	1&	80&	0.999983588&	1.64119E-05\\
35&	1.35466E-07&	0.999999865&	85&	0.999999949&	5.0732E-08\\
40&	1.80415E-05&	0.999981959&	90&	1&	2.33876E-11\\
45&	0.000881808&	0.999118192&	95&	1&	7.01609E-16\\
50&	0.016761687&	0.983238313&	100&	1&	6.53319E-23\\

\hline
\end{tabular}

\end{table}

\subsection{\goph for Image Near-duplicate Detection}
\label{sec:image detection}
In this subsection, we describe how we extend the \goph method originally developed
for text near-duplicate detection to image near-duplicate detection. We describe it using visual words to replace visual words in the following sub-section.

Two images are near duplicate if the similarity \textbf{sim($\mathcal{I}_1$,$\mathcal{I}_2$)}
is higher than a given similarity threshold $\mathcal{T}$. The goal is to retrieve all images
in the database that are similar to a query image. The outline of the images are near duplicate detection algorithm is as follows: First a list of visual words are extracted from each image. A visual word is a single number having the property that two images $\mathcal{I}_1$, $\mathcal{I}_2$ have the same value of visual word with probability equal to their similarity \textbf{sim($\mathcal{I}_1$,$\mathcal{I}_2$)}. To efficient compare the visual words, a one permutation $\psi$ is used to evenly divide these visual words into $k$-bins and retrieval smallest possible representations as fingerprints.
Then, we adopt \goph algorithm to compute the similarity between these two fingerprint sets.


\textbf{How does it work?} Consider image $\mathcal{Y}$ = arg min $\psi(\mathcal{I}_i \cup \mathcal{I}_j)$. Since $\psi$ is an one permutation, each fingerprint of $\mathcal{I}_i \cup \mathcal{I}_j$ has the same probability of being the smallest possible fingerprint. Hence, $\mathcal{Y}$ can be constructed from $\mathcal{I}_i \cup \mathcal{I}_j$. If $\mathcal{Y}$ is an fingerprint of both $\mathcal{I}_i$ and $\mathcal{I}_j$, i.e. $\mathcal{Y}$ $\subseteq$ $\mathcal{I}_i \cap \mathcal{I}_j$, then min $\psi(\mathcal{I}_i)$ = $\psi(\mathcal{I}_j)$ = $\psi(\mathcal{Y})$. Otherwise, if $\mathcal{Y}$ $\subseteq$ $\mathcal{I}_i \setminus \mathcal{I}_j$, then $\psi(\mathcal{Y})$ $<$ $\psi(\mathcal{I}_j)$; if $\mathcal{Y}$ $\subseteq$ $\mathcal{I}_j \setminus \mathcal{I}_i$, then $\psi(\mathcal{Y})$ $<$ $\psi(\mathcal{I}_i)$. Thus, for an one permutation $\psi$ it follows

\begin{equation}
\textbf{Pr}(min(\psi(\mathcal{I}_i))= min(\psi(\mathcal{I}_i))) = \dfrac{\mathcal{I}_i \cap \mathcal{I}_j}{\mathcal{I}_i \cup \mathcal{I}_j}
\end{equation}

To enhance the efficiency of comparison, the fingerprints are grouped into $m$-tuples. 
Similar to text comparison, from the first summary to the last summary, we gradually compute the similarity between the corresponding summaries, and estimate whether the remain summaries will meet the small probability event or not. If the remain summaries meets the small probability event, the algorithm terminals; Otherwise, we continue to calculate the fingerprints of the next summary.


\section{Our Strategy: Hierarchical One Permutation Hashing}
\label{sec:hierarchical one hash}
This section presents the Hierarchical One Permutation Hash
(\hoph) for efficient multi-modal near duplicate detection. Section ~\ref{sec:hoph overview} first
explains the basic idea of \hoph. Then, section ~\ref{sec:theoretical analysis}
presents some theoretical analysis of \hoph,
and section ~\ref{sec:hoph image} discuss the image implementation of \hoph.

\subsection{Hierarchical One Permutation Hashing Overview}
\label{sec:hoph overview}

Different with traditional \oph, which evenly divided the whole space $\mathcal{V}$ into $k$ buckets, \hoph scheme is first grouping the original data entries into two groups, namely permutation group and division group, in each iteration. The space of each permutation group is evenly into $k'$ parts, while the space of each division group is further divide into permutation group and division group sequentially, if the sub group size is greater than $k'$. More specifically, assume $a+b=1$, \hoph divide the space into two groups with proportion $a:b$, namely \abhoph. Thus, after first iteration, the
the whole space $\mathcal{V}$ is dividing into two groups $\mathcal{G}_1$ and $\mathcal{G}_2$ with size $\dfrac{a}{a+b}\mathcal{V}$, $\dfrac{b}{a+b}\mathcal{V}$ respectively. Then, we divide the space evenly into $k'$ parts for the first group $\mathcal{G}_1$. For the second group $\mathcal{G}_2$, we check whether the number of data entries of its sub groups are larger than $k'$ after division, if all the number is larger than $k'$, then \hoph continue to divide the second group into two parts with $a:b$; Otherwise, \hoph terminal the division. Figure \ref{fig:hoph_example} shows the example of \oohoph.

\begin{figure}[thb]
\newskip\subfigtoppskip \subfigtopskip = -0.1cm
\centering
\includegraphics[width=.90\linewidth]{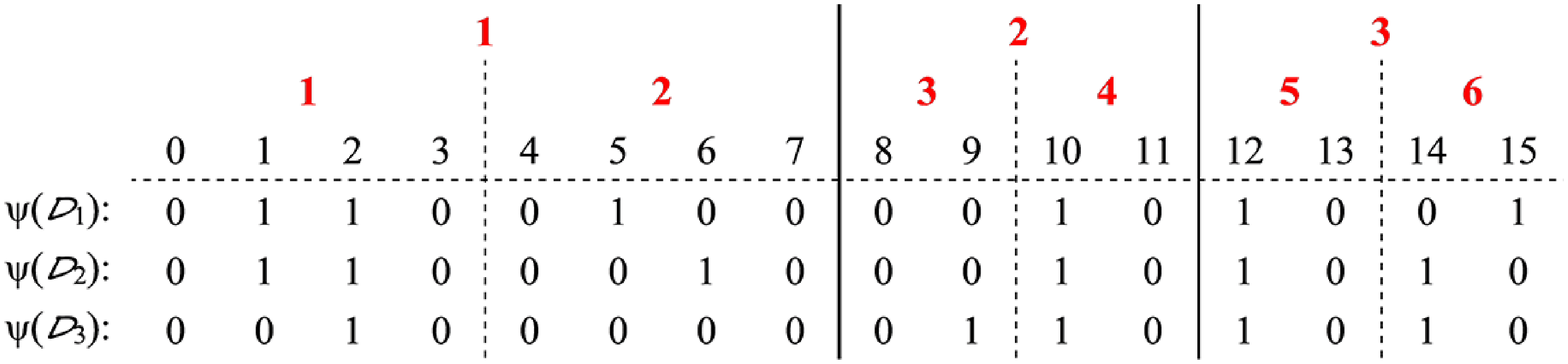}
\vspace{-1mm}
\caption{\small Hierarchical One Permutation Hashing Example }
\label{fig:hoph_example}
\end{figure}

Similar to \oph, a random permutation $\psi$ is generated firstly.
For each document $\mathcal{D}_i$ a one permutation hashing min $\psi(\mathcal{D}_i)$ is recorded. Consider $\mathcal{D}_1$, $\mathcal{D}_2$, $\mathcal{D}_3$ $\subseteq$ $\mathcal{V}$ = \{0, 1, ..., 16\}. Assume $\mathcal{D}_1$ = \{1, 2, 5, 10, 12, 15\}, $\mathcal{D}_2$ = \{1, 2, 6, 10, 12, 14\}, $\mathcal{D}_3$ = \{2, 9, 10, 12, 14\}, and $k'=2$ and $a:b=1:1$. Similar to \oph, a one permutation $\psi$ is generated firstly. Then, we divide the space with \oohoph into three groups. After that, we apply one permutation $\psi$ on the three groups and evenly divide the space into $k'$ buckets for each group, select the smallest nonzero in each bucket as samples. For example, [[1, 5], [$\otimes$, 10], [12, 15]], [[1, 6], [$\otimes$, 10], [12, 14]], and [[2, $\otimes$], [9, 10], [12, 14]]. We use $\otimes$ to denote an empty bin, which occur rarely while the number of nonzeros is large compared to $k'$.



In the example in Figure ~\ref{fig:hoph_example}(which includes 3 documents), the sample selected from $\psi(\mathcal{D}_1)$ is [[1, 5], [$\otimes$, 10], [12, 15]]. We re-index the elements of each bucket to use the smallest possible representations, because only elements with the same bin number need to be compared. For example, for $\psi(\mathcal{D}_1)$, after re-indexing, the sample [[1, 5], [$\otimes$, 10], [12, 15]] becomes [[1, 1], [$\otimes$, 0], [0, 1]]. Similarly, for $\psi(\mathcal{D}_2)$, the original sample [[1, 6], [$\otimes$, 10], [12, 14]] becomes [[1, 2], [$\otimes$, 0], [0, 0]], etc.

From the above example, the sets $\psi(\mathcal{D}_1)$ and $\psi(\mathcal{D}_2)$ have one identical smallest possible representations in each subgroup and the estimated similarity will be 0.625, while the exact similarity is 0.5. The sets $\psi(\mathcal{D}_1)$ and $\psi(\mathcal{D}_3)$ share one smallest possible representation and their similarity estimate is 0.375 (0.375 is exact).

\textbf{Algorithm} Algorithm \ref{alg:abhoph comp} illustrates the implementation details of the \abhoph Comparison for the input document pair. For presentation
simplicity, we assume there are $l$ groups in each document, $r=\dfrac{a}{a+b}$, $\mathcal{P}_a$ represents the average matched probability should be achieve, $\mathcal{P}_r$ represents the excepted average matched probability for the groups haven't been compared, and an array $\mathcal{M}$  is used to store the number of matched bin for each group in Line \ref{alg:abhoph init}. For
each group of document pair $(\mathcal{D}_i, \mathcal{D}_j)$, we gradually calculate their similarity. From Line
\ref{alg:abhoph scan st} to Line \ref{alg:abhoph scan end}, we update the $\mathcal{M}_l$ based on the $l$-th group's information. Then, we compute
the value of $P_r$ in Line \ref{alg:abhoph p_r}, we compare it with $\mathcal{P}_a$. From Line \ref{alg:abhoph sm st} to Line
\ref{alg:abhoph sm end}, if $P_r < P_a$, we further compare $F(P_r*k')$ with $\epsilon$, if its value is larger than $\epsilon$,
we continue to compare the similarity of next group; Otherwise, we consider current
document pair is a similar pair, and break the algorithm. From Line \ref{alg:abhoph la st} to Line \ref{alg:abhoph la end}, If
$P_r \geq P_a$, and $F(P_r*k')$ is larger than $\epsilon$, we continue to compare the similarity of
next group; Otherwise, we break the algorithm and consider current document pair is
not similar. Following is an example of \hoph comparison algorithm.

\begin{example}
\label{ex:motivation}
Given a document pair $(\mathcal{D}_i, \mathcal{D}_j)$, assume $k'$=100, $n$=10, $\mathcal{T}$=0.6, $\epsilon={10}^{-4}$ and $\mathcal{X} \sim F(\mathcal{T},\mathcal{K})$. Then, the values of $F(\mathcal{X})$ for different $\mathcal{X}$ are shown in Table \ref{tab:prob} and the value of $\mathcal{M}_a$ is 60. Assume there are only 40 bins matching in the 1-st
comparison, then the value of $P_r$ is 0.8. Because $P_r \geq P_a$, we further compare the value of $F(P_r*k')$ and $\epsilon$. Because $P_r*k'$ is 80, and the probability of $F(x\geq80)$ equals 1.64119E-05, which is smaller than $\epsilon$. It indicates the event that the document pair $(\mathcal{D}_1, \mathcal{D}_2)$ is a similar pair is a small probability event. Thus, we terminal the algorithm and consider current pair is not similar pair.
\end{example}

In Algorithm \ref{alg:abhoph comp}, line 3-8 computes the number of the bins which have equal value, and the cost is $O(k')$. Line 11-25 searches the similar document pairs by probability $Pr$ and $F(Pr*k')$, the cost is $O(1)$. Assume $F(Pr*k')=\epsilon$, the inverse function of $F$ is denoted as $F^{-1}$, thus $Pr=\frac{F^{-1}(\epsilon)}{k'}$. According to the aforementioned theories, $Pr=\frac{P-\frac{\mathcal{M}_l}{k'}*r^l}{r^l}=\frac{P}{r^l}-\frac{\mathcal{M}_l}{k'}$, we can compute $l=\log_r\frac{P*k'}{F^{-1}(\epsilon)+\mathcal{M}_l}$. Therefore, the total cost of algorithm \ref{alg:abhoph comp} is $O(k'*\log_r\frac{P*k'}{F^{-1}(\epsilon)+\mathcal{M}_l})$.

\begin{algorithm}[h]
\begin{algorithmic}[1]
\footnotesize
\caption{\bf \abhoph Comparison}
\label{alg:abhoph comp}

\INPUT  \\
$(\mathcal{D}_i, \mathcal{D}_j)$: the input document pair;
$\mathcal{T}$: user preferred similarity threshold, $k'$: bin number of each group;
$\mathcal{M}$: array to store the number of matched bin ;
$\epsilon$: error tolerance;
\\

\OUTPUT \\
$\mathcal{O}$: Similar document pairs\\

\STATE $l$=0; $\mathcal{M} \leftarrow \{0\}$; $\mathcal{P}$=$\mathcal{P}_a$=$\mathcal{T}$ $P_r=0$; $r=\dfrac{a}{a+b}$;
\label{alg:abhoph init}
\FOR{each $l$-th group of document pair $(\mathcal{D}_i, \mathcal{D}_j)$}
    \FOR{each bins in $l$-th group}
    \label{alg:abhoph scan st}
    \STATE $\mathcal{M}_{l}$=0;
        \IF{the bin has equal value}
            \STATE $\mathcal{M}_l$++;
        \ENDIF
    \ENDFOR
    \label{alg:abhoph scan end}
    \STATE $l$++;
    \STATE $\mathcal{P}_r=\dfrac{\mathcal{P}-\dfrac{\mathcal{M}_l}{k'}r^l}{r^l}$;
    \STATE $\mathcal{P}=\mathcal{P}_r$;
    \label{alg:abhoph p_r}
    \IF{$\mathcal{P}_r < \mathcal{P}_a$}
    \label{alg:abhoph sm st}
    \IF{$F(\mathcal{P}_r*k')>\epsilon$}
            \STATE $Continue$;
        \ELSE
            \STATE $\mathcal{O} \leftarrow (\mathcal{D}_i, \mathcal{D}_j)$ ;
            \STATE $Break$;
        \ENDIF
    \ENDIF
    \label{alg:abhoph sm end}
    \IF{$\mathcal{P}_r \geq \mathcal{P}_a$}
    \label{alg:abhoph la st}
        \IF{$F(P_r*k')>\epsilon$}
            \STATE $Continue$;
        \ELSE
            \STATE $Break$;
        \ENDIF
    \ENDIF
    \label{alg:abhoph la end}
\ENDFOR

\end{algorithmic}
\end{algorithm}

\subsection{Theoretical Analysis of \hoph}
\label{sec:theoretical analysis}
In the following section, we will introduce some interesting theoretical analysis of \hoph, such as the number of match bin, the number of empty bin, the unbiased estimator
and so on. 



\begin{figure}[thb]
\newskip\subfigtoppskip \subfigtopskip = -0.1cm
\centering
\includegraphics[width=.90\linewidth]{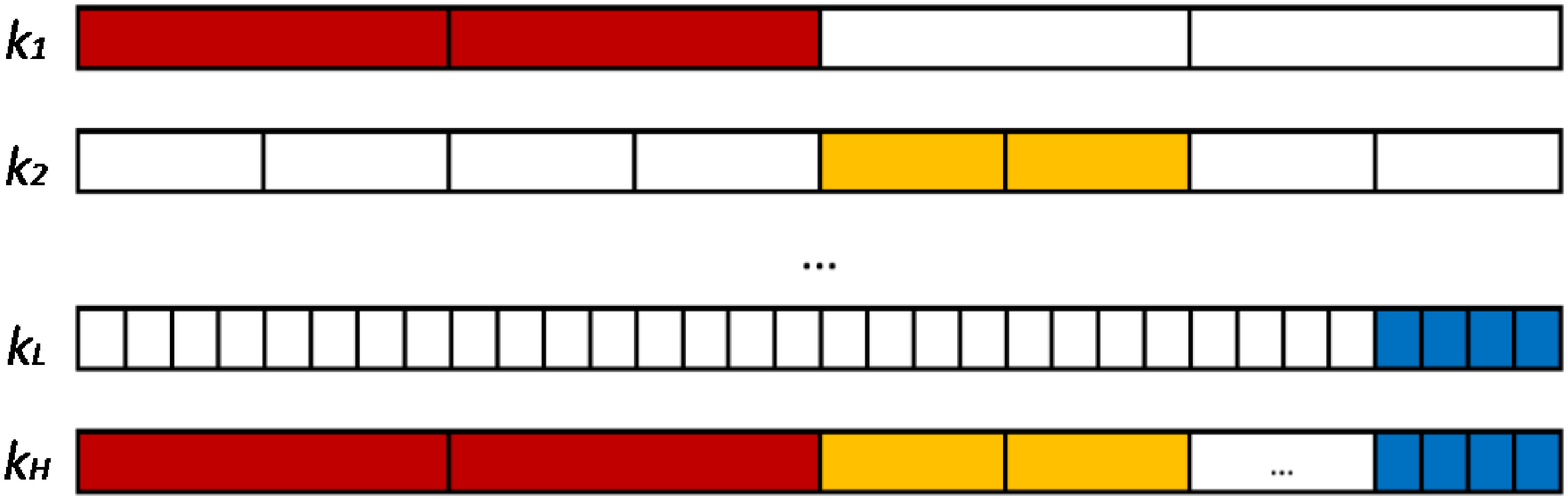}
\vspace{-1mm}
\caption{\small Hierarchical One Permutation Hashing Construction }
\label{fig:hoph proof_example}
\end{figure}

Assume $a+b=1, r=\dfrac{a}{a+b}$,  $(a:b)$\hoph first divides the whole space $\mathcal{V}$ into $l+1$ group with length $r\mathcal{V}$, $r^{2}\mathcal{V} \cdots r^{l}\mathcal{V}$ and $r^{l-1}(1-r)\mathcal{V}$ respectively, then divide the corresponding space evenly into $k'$ bins, as what is shown in Fig \ref{fig:hoph proof_example}. Based on the above assumption, $(a:b)$\hoph has the following properties:

\begin{lemma} \label{lemma:hoph mb}
The number of matched bin of $(a:b)$\hoph is
\begin{equation*}
\label{eqn:hoph mb}
\mathcal{N}_{mb_{h}}=(\sum_{j=1}^{l-1}r^{2j}\mathcal{N}_{mb_{j}}+r^{2l-1}\mathcal{N}_{mb_{l}})(l+1)
\end{equation*}
\end{lemma}

\begin{proof}
As shown in Fig \ref{fig:hoph proof_example}, assume $(a:b)$\hoph consists of $l+1$ groups from $l$ different one permutation hashes, which are evenly divided into $k_1, k_2, \cdots, k_l$ respectively.
\begin{equation*}
\begin{split}
P(\mathcal{B}_{mb,i}= 1, i \in [1, &k_{h}])=P(\mathcal{B}_{mb,i}=1,i \in [1,r{k_1}])\\
&+P(\mathcal{B}_{mb,i}=1,i \in [r{k_1}+1,r{k_1}+r^{2}{k_2}])+\cdots\\
&+P(\mathcal{B}_{mb,i}=1,i \in
[\sum_{j=1}^{l-1}r^{j}{k_{j}}+1,\sum_{j=1}^{l}r^{j}{k_{j}}])\\
&+P(\mathcal{B}_{mb,i}=1,i \in
[\sum_{j=1}^{l}r^{j}{k_{j}}]+1,\sum_{j=1}^{l}r^{j}{k_{j}}+r^{l-1}(1-r)k_{l}])\\
&=r\dfrac{\mathcal{N}_{mb_1}}{k_1}+r^{2}\dfrac{\mathcal{N}_{mb_2}}{k_2}+\cdots+r^{l}\dfrac{\mathcal{N}_{mb_l}}{k_l}
+r^{l-1}(1-r)\dfrac{\mathcal{N}_{mb_l}}{k_l}
\end{split}
\end{equation*}
\ \ \ \ Since,
\begin{equation*}
\begin{split}
P(\mathcal{B}_{mb,i}= 1, i \in [1, &k_{h}])=\dfrac{\mathcal{N}_{mb_{h}}}{k_{h}}
\end{split}
\end{equation*}

Then,
\begin{equation*}
\dfrac{\mathcal{N}_{mb_{h}}}{k_{h}}=r\dfrac{\mathcal{N}_{mb_1}}{k_1}+r^{2}\dfrac{\mathcal{N}_{mb_2}}{k_2}+\cdots+r^{l-1}\dfrac{\mathcal{N}_{mb_{l-1}}}{k_{l-1}}
+r^{l-1}\dfrac{\mathcal{N}_{mb_l}}{k_l}
\end{equation*}

Since,
\begin{equation*}
k_{h}=(l+1)k'; k_{j}=\dfrac{k'}{r^j}
\end{equation*}

We obtain,
\begin{equation*}
\begin{split}
\mathcal{N}_{mb_{h}}&=(r^{2}\dfrac{\mathcal{N}_{mb_1}}{k'}+r^{4}\dfrac{\mathcal{N}_{mb_2}}{k'}+\cdots+r^{2(l-1)}\dfrac{\mathcal{N}_{mb_{l-1}}}{k'}
+r^{2l-1}\dfrac{\mathcal{N}_{mb_l}}{k'})(l+1)k'\\
&=(r^{2}\mathcal{N}_{mb_1}+r^{4}\mathcal{N}_{mb_2}+\cdots+r^{2(l-1)}\mathcal{N}_{mb_{l-1}}+r^{2l-1}\mathcal{N}_{mb_{l}})(l+1)\\
&=(\sum_{j=1}^{l-1}r^{2j}\mathcal{N}_{mb_{j}}+r^{2l-1}\mathcal{N}_{mb_{l}})(l+1)
\end{split}
\end{equation*}

thus completing the proof.

\end{proof}

\begin{lemma}
The number of empty bin of $(a:b)$\hoph is
\begin{equation*}
\mathcal{N}_{eb_{h}}=(\sum_{j=1}^{l-1}r^{2j}\mathcal{N}_{eb_{j}}+r^{2l-1}\mathcal{N}_{eb_{l}})(l+1)
\end{equation*}
\end{lemma}
\begin{proof}
Similar to Lemma \ref{lemma:hoph mb}, we can obtain
\begin{equation*}
\begin{split}
P(\mathcal{B}_{eb,i}= 1, i \in [1, &k_{h}])=r\dfrac{\mathcal{N}_{eb_1}}{k_1}+r^{2}\dfrac{\mathcal{N}_{eb_2}}{k_2}+\cdots+r^{l}\dfrac{\mathcal{N}_{eb_l}}{k_l}
+r^{l-1}(1-r)\dfrac{\mathcal{N}_{eb_l}}{k_l}
\end{split}
\end{equation*}

\ \ \ \ Since,
\begin{equation*}
\begin{split}
P(\mathcal{B}_{eb,i}= 1, i \in [1, &k_{h}])=\dfrac{\mathcal{N}_{eb_{h}}}{k_{h}}; k_{h}=(l+1)k'; k_{j}=\dfrac{k'}{r^j}
\end{split}
\end{equation*}

We obtain,
\begin{equation*}
\begin{split}
\mathcal{N}_{eb_{h}}&=(r^{2}\dfrac{\mathcal{N}_{eb_1}}{k'}+r^{4}\dfrac{\mathcal{N}_{eb_2}}{k'}+\cdots+r^{2(l-1)}\dfrac{\mathcal{N}_{eb_{l-1}}}{k'}
+r^{2l-1}\dfrac{\mathcal{N}_{eb_l}}{k'})(l+1)k'\\
&=(\sum_{j=1}^{l-1}r^{2j}\mathcal{N}_{eb_{j}}+r^{2l-1}\mathcal{N}_{eb_{l}})(l+1)
\end{split}
\end{equation*}

thus completing the proof.
\end{proof}

\begin{lemma}
The unbiased estimator of $(a:b)$\hoph is
\begin{equation*}
\hat{\mathcal{R}}_{mb_{h}}=\sum_{j=1}^{l-1}r^{j}\dfrac{\mathcal{N}_{mb_{j}}}{k_{j}-\mathcal{N}_{eb_{j}}}+r^{l-1}\dfrac{\mathcal{N}_{mb_{l}}}{k_{l}-\mathcal{N}_{eb_{l}}}
\end{equation*}
\end{lemma}
\begin{proof}
Similar to Lemma \ref{lemma:hoph mb}, we can obtain
\begin{equation*}
\begin{split}
\hat{\mathcal{R}}_{mb_{h}}&=r\hat{\mathcal{R}}_{mb_{1}}+r^{2}\hat{\mathcal{R}}_{mb_{2}}+\cdots
+r^{l}\hat{\mathcal{R}}_{mb_{l}}+r^{l-1}(1-r)\hat{\mathcal{R}}_{mb_{l}}\\
&=\sum_{j=1}^{l-1}r^{j}\hat{\mathcal{R}}_{mb_{j}}+r^{l-1}\hat{\mathcal{R}}_{mb_{l}}
\end{split}
\end{equation*}

Since,
\begin{equation*}
\hat{\mathcal{R}}_{mb_{j}}=\dfrac{\mathcal{N}_{mb_{j}}}{k_{j}-\mathcal{N}_{eb_{j}}}
\end{equation*}

We obtain,
\begin{equation*}
\begin{split}
\hat{\mathcal{R}}_{mb_{h}}=\sum_{j=1}^{l-1}r^{j}\dfrac{\mathcal{N}_{mb_{j}}}{k_{j}-\mathcal{N}_{eb_{j}}}+r^{l-1}\dfrac{\mathcal{N}_{mb_{l}}}{k_{l}-\mathcal{N}_{eb_{l}}}
\end{split}
\end{equation*}

thus completing the proof.
\end{proof}

\subsection{\hoph for Image Near-duplicate Detection}
\label{sec:hoph image}
In term of image near duplicate detection, the construction and comparison method of \hoph is similar to that of \goph. Given a image collection, we first extract the visual word list from each image. In construction stage, we first generate a random permutation $\psi$, then apply \hoph scheme recursively divide the whole visual word space two groups in each iteration. For the front group, we apply  permutation $\psi$ to it, and evenly divide the space into $k'$ buckets. For the latter group, we further divide the space into two groups again, if its sub group size is greater than $k'$. In comparison stage, given two \hoph group, we gradually compute the similarity between the corresponding group from the first to the last, and estimate
whether the remain part will meet the small probability event or not. If the remain part will trigger the small probability event, the algorithm terminals and outputs the result; Otherwise, the fingerprints of the next group will be calculated for further evaluation.


\section{PERFORMANCE EVALUATION}
\label{perform}

In this section, we present results of a comprehensive performance study to evaluate the efficiency and scalability of the proposed techniques in the paper. In our implementation, we evaluate the effectiveness of the following Hashing techniques.

\begin{itemize}
\item \minwise. Minwish hashing, which is a natural implementation of the method in ~\cite{DBLP:journals/jcss/BroderCFM00}.
\item \oph. One permutation hashing, which is a natural implementation of the method in ~\cite{Li2012One}.
\item \goph. The group based one permutation hashing technique proposed in Section ~\ref{sec:GOPH}.
\item \oohoph. The hierarchical one permutation hashing whose ratio of $a$ to $b$ equals 1:1.
\item \othoph. The hierarchical one permutation hashing whose ratio of $a$ to $b$ equals 1:2.
\item \tohoph. The hierarchical one permutation hashing whose ratio of $a$ to $b$ equals 2:1.

\end{itemize}

\textbf{Environment Settings.} Experiments are run on a PC with Intel i7 6700HQ 2.60GHz CPU and 16G memory running Ubuntu 16.04 LTS. All algorithms in the experiments are implemented in Java.

\textbf{Workload.} A workload for this experiment consists of 100 input queries, and the precision, recall and response time are employed to evaluate the performance of the algorithms. By default, we set the error tolerance $e = 10^{-4}$, user preferred similarity threshold $T = 0.7$, data number $V = 60 * 10^4$.

\textbf{Performance matric.} The objective evaluation of the proposed approach is carried out based on precision and recall. Precision measures the accuracy of the retrieval. It is the ratio of retrieved documents that are similar to the query.

Precision is the ratio of retrieved images that are relevant to the query image.

\begin{displaymath}
Precision = \dfrac{number \quad of \quad similar \quad documents \quad retrieved}{Total \quad number \quad of \quad documents \quad retrieved}
\end{displaymath}

Recall measures the robustness of the retrieval. It is defined as the ratio of relevant images in the database that are retrieved in response to a query.

\begin{displaymath}
Recall = \dfrac{number \quad of \quad similar \quad documents \quad retrieved}{Total \quad number \quad of \quad similar \quad documents \quad in \quad dataset}
\end{displaymath}

\subsection{Evaluation on Text Dataset (FS)}
\textbf{Dataset.} Performance of various algorithms are evaluated on real dataset FundSet(FS). FS is obtained from the large-scale document database of NSFC in which each document is a NSFC proposal in PDF format. Taking some documents of funds proposal as the data source.

\textbf{Evaluating training group number.} At first, we try to train the group number of text document $n$. Fig.(a) shows that with $n$ increasing, response time is reduced. But the downtrend slows down and at $n = 10$ the response time is minimum. It indicates that the performance cannot be boosted all along with the gradually increasing of $n$. On the other hand, as shown in Fig.(b), no matter what value $n$ is, the precision is very high, nearly 100\%. Thus, we choose the $n =10$ as the default value of GOPH.

\textbf{Evaluation on different data number.} We investigate the response time, precision and Recall in Fig. against the dataset FS, where other parameters are set to default values. Fig.(a) depicts the accuracy of GOPH, 1:2HOPH, 1:1HOPH, 2:1HOPH, MinHash and OPH. Obviously, MinHash and OPH have the highest precision. When Data Num is larger than $60 * 10^4$, the growth of precision slows down. Fig.(b) demonstrates that the recall of them all are climbing with the increasing of Data Num, and there is little difference between them. Fig.(c) shows that with the increasing of the Data Num, the response time of these 6 algorithms gradually rise. At Data Num $= 20 * 10^4$, all of the values are in the range between 150 and 250. But when Data Num increases to 100, we can see that the performance of 1:1HOPH and 2:1HOPH are much better than MinHash and OPH. Particularly, 2:1HOPH has the smallest response time among the algorithms, because the number of filter segments is the most in 2:1HOPH. On the other hand, 1:1HOPH is most suitable because of the equilibrium of precision, recall and speed. As above evaluation shown, in the aspect of efficiency 1:1HOPH is higher than 1:2HOPH. Besides, the response time of 1:1HOPH and 2:1HOPH are almost the same but the precision of 1:1HOPH is higher than the other. Meanwhile, as the accuracy of \goph is close to \oph and its response time is two or three times faster than \oph, we select GOPH and 1:1HOPH to conduct the following comparison evaluation.

\begin{figure*}
\newskip\subfigtoppskip \subfigtopskip = -0.1cm
\begin{minipage}[b]{0.99\linewidth}
\begin{center}
     \subfigure[{Precision}]{
     \includegraphics[width=0.31\linewidth]{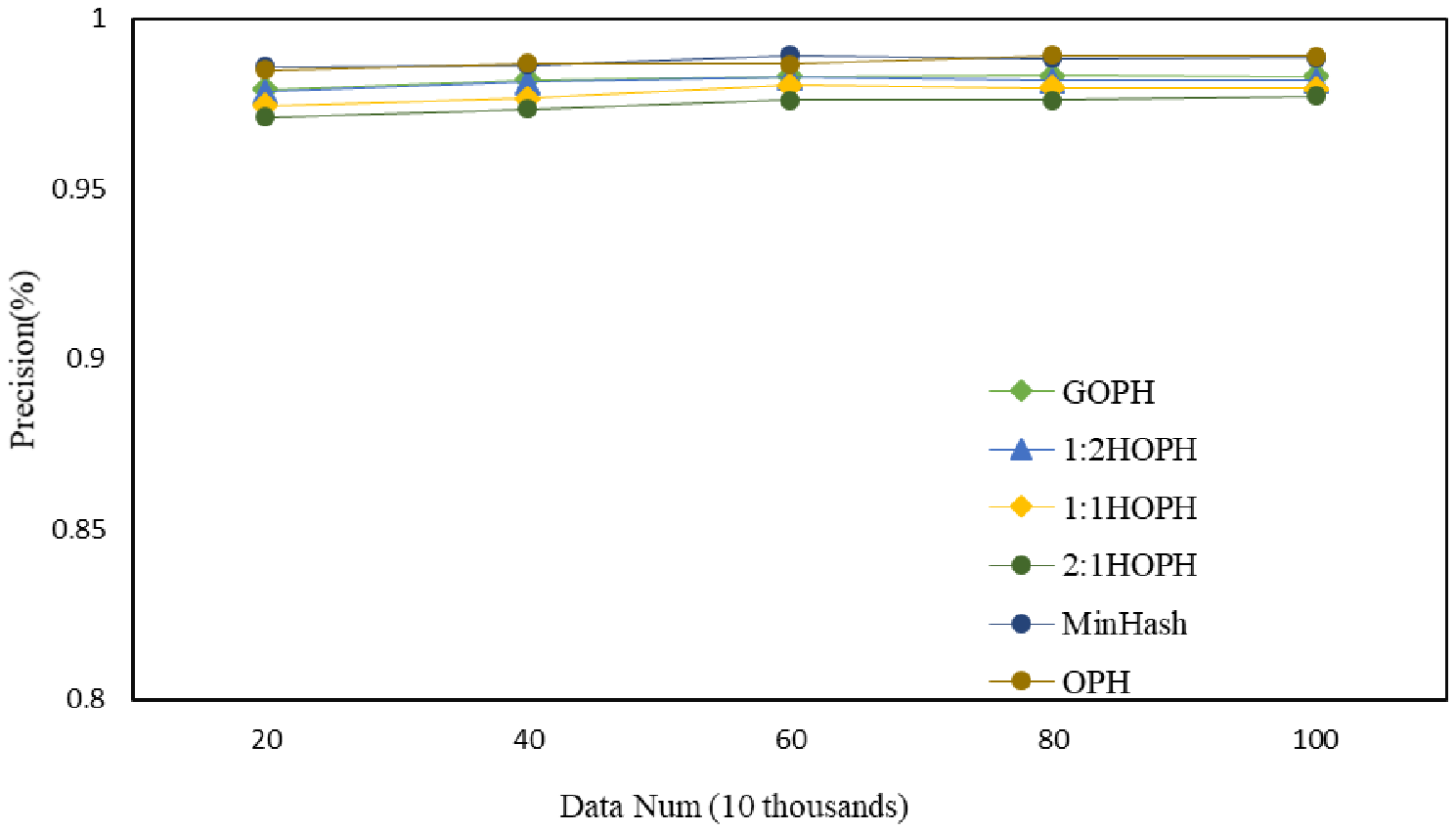}
     }
     \subfigure[{Response time}]{
     \includegraphics[width=0.31\linewidth]{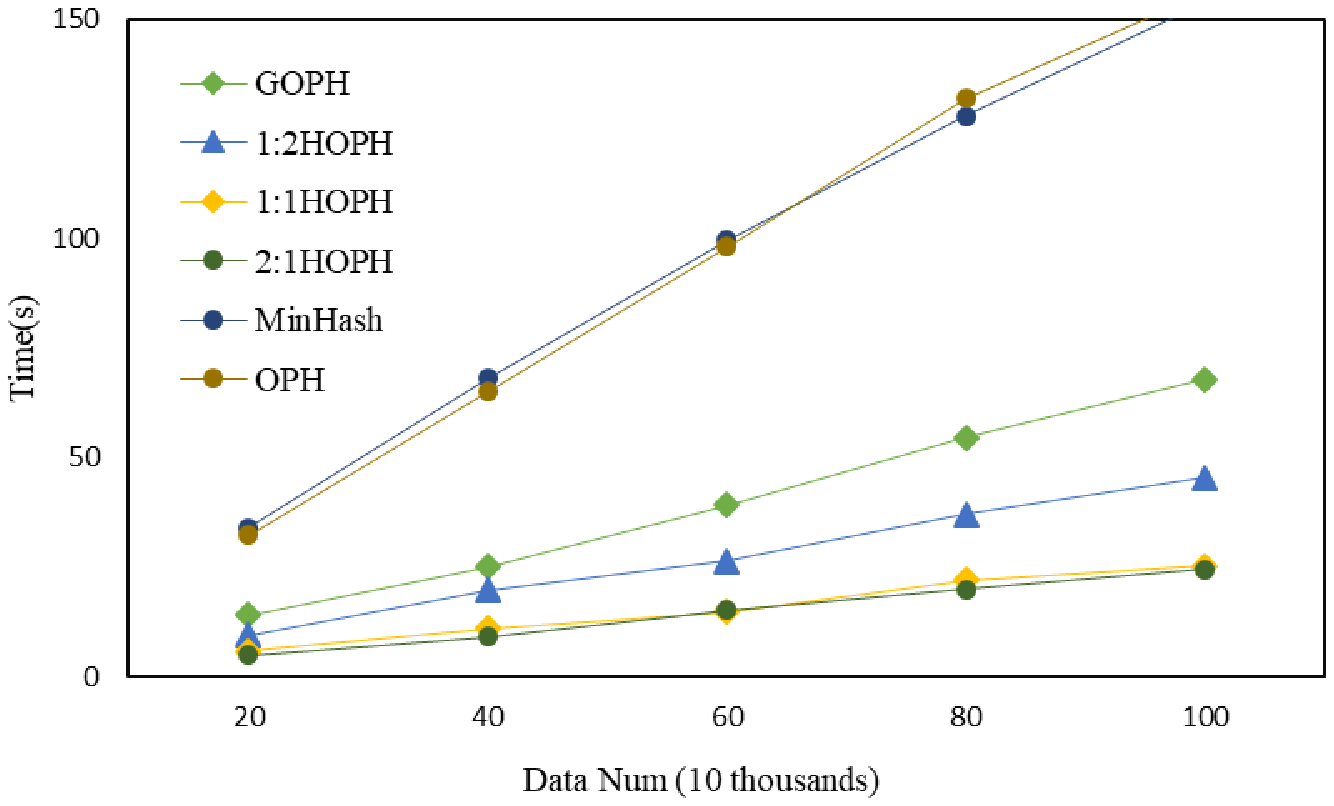}
     }
   \captionsetup{justification=centering}
       \vspace{-0.2cm}
\caption{Effect of different data number on FS}
\label{fig:Fig2}
\end{center}
\end{minipage}
\label{fig:k}
\end{figure*}

\textbf{Evaluation on different threshold.} Fig.(a) depicts that with the rising of the threshold, the precision of MinHash, GOPH and 1:1HOPH slowly increases. All of them are larger than 98\% when the threshold is larger 0.7. Apparently, the precision of MinHash is little higher than the precision of 1:1HOPH. Fig.(b) illustrates that the three recalls stay the same tendency, they are near 100\% when the threshold $T$ is larger than 0.8. As shown in Fig.(c), GOPH, MinHash and 1:1HOPH demonstrate superior performance in comparison with MinHash and the response time of GOPH and 1:1HOPH decline with the Threshold climbing. Obviously, compared with MinHash and GOPH, 1:1HOPH has the smallest response time.

\textbf{Evaluation on Small Probability.} Fig.(a) reports the precision of Minhash, GOPH and 1:1HOPH. Clearly, the precision of MinHash is invariable and the precision of GOPH and 1:1HOPH grow slowly. All of them are very high, MinHash is nearly 99\% and two others are larger than 98\%. Fig.(b) demonstrates the recall of GOPH, MinHash and 1:1HOPH. It is easy to find there is little difference in recall. All of them gradually increase and are nearly 100\% when the error tolerance equals 1.00E-05. In Fig.(c), the performance of 1:1HOPH is nearly 5 times higher than MinHash. When we change the error tolerance to a smaller value, the climbing of the performance of GOPH and 1:1HOPH is not obvious.

\subsection{Evaluation on Image Dataset (IS)}
\textbf{Dataset.} Our empirical studies aim to evaluate the performance of the filter against a subset of ImageNet datast. 
ImageNet is the largest image dataset for image processing and computer vision. It is organized according to the WordNet hierarchy (currently only the nouns), in which each node of the hierarchy is depicted by hundreds and thousands of images. This dataset includes: (1)14,197,122 images; (2)1,034,908 images with bounding box annotations; (3)1000 synsets with SIFT features; (4)1.2 million images with SIFT features. We generate a image dataset (IS) by selecting 1 million images from ImageNet. On IS, we evaluate the precision and efficiency of GOPH, 1:2HOPH, 1:1HOPH, 2:1HOPH, MinHash and OPH.

\textbf{Evaluating training group number.} Firstly, we try to train the group number $n$. It is illustrated by Fig.(a) that with the raising of $n$, response time is going down gradually. But the downtrend slows down and at $n = 10$ the response time is minimum. The performance cannot be boosted all along with $n$ gradually increasing. On the other hand, as shown in Fig.(b), with the raising of Data Num, the precision increases step by step with fluctuations, and at $n = 10$, it is nearly 99.4\%. Hence we choose the $n =10$ as the default value of GOPH.

\textbf{Evaluation on different data number.} We evaluate the precision, recall and response time of these 6 algorithms against IS. The precision is shown in Fig.(a). apparently, all the precisions fluctuate in the range of 97.1\% to 98.7\% with the increasing of Data Num and the precision of MinHash and OPH are higher than others. It is obvious that in Fig.(b) there is litter difference in recall over these 6 algorithms and all of them are approximate 100\% with the number of data increasing. Fig.(c) demonstrates the trends of response time of MinHash and OPH are almost the same, much higher than the others. Particularly, 2:1HOPH and 1:1HOPH significantly outperform the other 4 algorithms in performance. It is clear that 1:1HOPH dominates to 2:1HOPH in the aspect of precision but the efficiency of the former is not lower than the the latter. On the other hand, the response time of 1:2HOPH is higher than 1:1HOPH. Furthermore, the efficiency of GOPH is much higher than OPH. Therefore, in the evaluation on different threshold and small PR, we compare the two algorithms mentioned-above and MinHash.

\textbf{Evaluation on different threshold.} The precision of GOPH, MinHash and 1:1HOPH on difference threshold are shown in Fig.(a). All of the precisions fluctuate in the interval of 97.9\% to 98.7\% The precision of MinHash is little higher than 1:1HOPH and GOPH. As shown in Fig.(b), the recall of these algorithms stay the same tendency. They ascend gradually when the Threshold increases from 0.5 to 0.7. Fig.(c) tells us that the response time of GOPH and 1:1HOPH decline step by step but the performance of MinHash remains unchanged. As expected, 1:1HOPH has the best performance among them. When the Threshold is smaller than 0.8, the response time of 1:1HOPH is less than $13s$.

\textbf{Evaluation on Small Probability.} In Fig.(a), we evaluate the precision of GOPH, MinHash and 1:1HOPH. It is no doubt that the precision of MinHash stay a very high value, a little higher than the others which slowly raise with the error tolerance increasing from 1.00E-3 to 1.00E-4. After that they are almost invariable.  On the other hand, as shown in Fig.(b) the recall of these algorithms go up moderately and at error tolerance is 1.00E-5 they are nearly 100\%. We can see from Fig.(c) that, with the changing of error tolerance, the performance of MinHash remains the same but the others changed very smoothly. Like the situation on dataset FS, the performance of 1:1HOPH is much better than two others.

\section{Conclusion}
\label{con}
The problem of multimedia near duplicate detection is important
due to the increasing amount of multimedia data collected
in a wide spectrum of applications. In the paper, we propose introduce \oph to reduce the costly preprocessing time. Based on \oph, we propose \goph to accelerate the comparison speed. Then, we design a novel hashing method namely \hoph to further improve the performance. Both \goph and \hoph can easily extend to image near duplicate detection.
Finally, our comprehensive experiments
convincingly demonstrate the efficiency of our techniques.

\textbf{Acknowledgments:} This work was supported in part by the National Natural Science Foundation of China
(61379110, 61472450, 61702560), the Key Research Program of Hunan Province(2016JC2018), and project (2016JC2011, 2018JJ3691) of Science and Technology Plan of Hunan Province.

\bibliographystyle{spmpsci}      
\bibliography{ref}

\begin{thebibliography}{10}
\providecommand{\url}[1]{{#1}}
\providecommand{\urlprefix}{URL }
\expandafter\ifx\csname urlstyle\endcsname\relax
  \providecommand{\doi}[1]{DOI~\discretionary{}{}{}#1}\else
  \providecommand{\doi}{DOI~\discretionary{}{}{}\begingroup
  \urlstyle{rm}\Url}\fi

\bibitem{DBLP:conf/www/BayardoMS07}
Bayardo, R.J., Ma, Y., Srikant, R.: Scaling up all pairs similarity search.
\newblock In: Proceedings of the 16th International Conference on World Wide
  Web, {WWW} 2007, Banff, Alberta, Canada, May 8-12, 2007, pp. 131--140 (2007)

\bibitem{DBLP:journals/jcss/BroderCFM00}
Broder, A.Z., Charikar, M., Frieze, A.M., Mitzenmacher, M.: Min-wise
  independent permutations.
\newblock J. Comput. Syst. Sci. \textbf{60}(3), 630--659 (2000)

\bibitem{DBLP:journals/cn/BroderGMZ97}
Broder, A.Z., Glassman, S.C., Manasse, M.S., Zweig, G.: Syntactic clustering of
  the web.
\newblock Computer Networks \textbf{29}(8-13), 1157--1166 (1997)

\bibitem{DBLP:conf/civr/ChumPIZ07}
Chum, O., Philbin, J., Isard, M., Zisserman, A.: Scalable near identical image
  and shot detection.
\newblock In: Proceedings of the 6th {ACM} International Conference on Image
  and Video Retrieval, {CIVR} 2007, Amsterdam, The Netherlands, July 9-11,
  2007, pp. 549--556 (2007)

\bibitem{DBLP:conf/bmvc/ChumPZ08}
Chum, O., Philbin, J., Zisserman, A.: Near duplicate image detection: min-hash
  and tf-idf weighting.
\newblock In: Proceedings of the British Machine Vision Conference 2008, Leeds,
  September 2008, pp. 1--10 (2008)

\bibitem{DBLP:journals/eswa/Hassanian-esfahani18}
Hassanian{-}esfahani, R., Kargar, M.J.: Sectional minhash for near-duplicate
  detection.
\newblock Expert Syst. Appl. \textbf{99}, 203--212 (2018)

\bibitem{DBLP:conf/sigir/Henzinger06}
Henzinger, M.R.: Finding near-duplicate web pages: a large-scale evaluation of
  algorithms.
\newblock In: {SIGIR} 2006: Proceedings of the 29th Annual International {ACM}
  {SIGIR} Conference on Research and Development in Information Retrieval,
  Seattle, Washington, USA, August 6-11, 2006, pp. 284--291 (2006)

\bibitem{DBLP:journals/jasis/HoadZ03}
Hoad, T.C., Zobel, J.: Methods for identifying versioned and plagiarized
  documents.
\newblock {JASIST} \textbf{54}(3), 203--215 (2003)

\bibitem{DBLP:conf/stoc/IndykM98}
Indyk, P., Motwani, R.: Approximate nearest neighbors: Towards removing the
  curse of dimensionality.
\newblock In: Proceedings of the Thirtieth Annual {ACM} Symposium on the Theory
  of Computing, Dallas, Texas, USA, May 23-26, 1998, pp. 604--613 (1998)

\bibitem{DBLP:conf/cvpr/JainKG08}
Jain, P., Kulis, B., Grauman, K.: Fast image search for learned metrics.
\newblock In: 2008 {IEEE} Computer Society Conference on Computer Vision and
  Pattern Recognition {(CVPR} 2008), 24-26 June 2008, Anchorage, Alaska, {USA}
  (2008)

\bibitem{Li2012One}
Li, P., Owen, A., Zhang, C.H.: One permutation hashing for efficient search and
  learning.
\newblock Mathematics  (2012)

\bibitem{DBLP:conf/nips/LiSMK11}
Li, P., Shrivastava, A., Moore, J.L., K{\"{o}}nig, A.C.: Hashing algorithms for
  large-scale learning.
\newblock In: Advances in Neural Information Processing Systems 24: 25th Annual
  Conference on Neural Information Processing Systems 2011. Proceedings of a
  meeting held 12-14 December 2011, Granada, Spain., pp. 2672--2680 (2011)

\bibitem{DBLP:journals/tmm/LiWCXL13}
Li, P., Wang, M., Cheng, J., Xu, C., Lu, H.: Spectral hashing with semantically
  consistent graph for image indexing.
\newblock {IEEE} Trans. Multimedia \textbf{15}(1), 141--152 (2013)

\bibitem{DBLP:conf/iccv/Lowe99}
Lowe, D.G.: Object recognition from local scale-invariant features.
\newblock In: {ICCV}, pp. 1150--1157 (1999)

\bibitem{DBLP:journals/ijcv/Lowe04}
Lowe, D.G.: Distinctive image features from scale-invariant keypoints.
\newblock International Journal of Computer Vision \textbf{60}(2), 91--110
  (2004)

\bibitem{DBLP:conf/cvpr/NisterS06}
Nist{\'{e}}r, D., Stew{\'{e}}nius, H.: Scalable recognition with a vocabulary
  tree.
\newblock In: 2006 {IEEE} Computer Society Conference on Computer Vision and
  Pattern Recognition {(CVPR} 2006), 17-22 June 2006, New York, NY, {USA}, pp.
  2161--2168 (2006)

\bibitem{DBLP:conf/pods/PaghSW14}
Pagh, R., St{\"{o}}ckel, M., Woodruff, D.P.: Is min-wise hashing optimal for
  summarizing set intersection?
\newblock In: Proceedings of the 33rd {ACM} {SIGMOD-SIGACT-SIGART} Symposium on
  Principles of Database Systems, PODS'14, Snowbird, UT, USA, June 22-27, 2014,
  pp. 109--120 (2014)

\bibitem{DBLP:conf/cvpr/PhilbinCISZ07}
Philbin, J., Chum, O., Isard, M., Sivic, J., Zisserman, A.: Object retrieval
  with large vocabularies and fast spatial matching.
\newblock In: 2007 {IEEE} Computer Society Conference on Computer Vision and
  Pattern Recognition {(CVPR} 2007), 18-23 June 2007, Minneapolis, Minnesota,
  {USA} (2007)

\bibitem{DBLP:conf/icimcs/QuSYL13}
Qu, Y., Song, S., Yang, J., Li, J.: Spatial min-hash for similar image search.
\newblock In: International Conference on Internet Multimedia Computing and
  Service, {ICIMCS} '13, Huangshan, China - August 17 - 19, 2013, pp. 287--290
  (2013)

\bibitem{DBLP:journals/prl/ShaoWOZ12}
Shao, J., Wu, F., Ouyang, C., Zhang, X.: Sparse spectral hashing.
\newblock Pattern Recognition Letters \textbf{33}(3), 271--277 (2012)

\bibitem{DBLP:conf/iccv/SivicZ03}
Sivic, J., Zisserman, A.: Video google: {A} text retrieval approach to object
  matching in videos.
\newblock In: 9th {IEEE} International Conference on Computer Vision {(ICCV}
  2003), 14-17 October 2003, Nice, France, pp. 1470--1477 (2003)

\bibitem{DBLP:conf/cvpr/TorralbaFW08}
Torralba, A., Fergus, R., Weiss, Y.: Small codes and large image databases for
  recognition.
\newblock In: 2008 {IEEE} Computer Society Conference on Computer Vision and
  Pattern Recognition {(CVPR} 2008), 24-26 June 2008, Anchorage, Alaska, {USA}
  (2008)

\bibitem{DBLP:journals/siamdm/Vsemirnov04}
Vsemirnov, M.: Automorphisms of projective spaces and min-wise independent sets
  of permutations.
\newblock {SIAM} J. Discrete Math. \textbf{18}(3), 592--607 (2004)

\bibitem{YangINF2013}
Wang, Y., Huang, X., Wu, L.: Clustering via geometric median shift over
  riemannian manifolds.
\newblock Information Sciences \textbf{220}, 292--305 (2013)

\bibitem{DBLP:conf/mm/WangLWZ15}
Wang, Y., Lin, X., Wu, L., Zhang, W.: Effective multi-query expansions: Robust
  landmark retrieval.
\newblock In: Proceedings of the 23rd Annual {ACM} Conference on Multimedia
  Conference, {MM} '15, Brisbane, Australia, October 26 - 30, 2015, pp. 79--88
  (2015)

\bibitem{DBLP:journals/tip/WangLWZ17}
Wang, Y., Lin, X., Wu, L., Zhang, W.: Effective multi-query expansions:
  Collaborative deep networks for robust landmark retrieval.
\newblock {IEEE} Trans. Image Processing \textbf{26}(3), 1393--1404 (2017)

\bibitem{DBLP:conf/mm/WangLWZZ14}
Wang, Y., Lin, X., Wu, L., Zhang, W., Zhang, Q.: Exploiting correlation
  consensus: Towards subspace clustering for multi-modal data.
\newblock In: Proceedings of the {ACM} International Conference on Multimedia,
  {MM} '14, Orlando, FL, USA, November 03 - 07, 2014, pp. 981--984 (2014)

\bibitem{DBLP:conf/sigir/WangLWZZ15}
Wang, Y., Lin, X., Wu, L., Zhang, W., Zhang, Q.: {LBMCH:} learning bridging
  mapping for cross-modal hashing.
\newblock In: Proceedings of the 38th International {ACM} {SIGIR} Conference on
  Research and Development in Information Retrieval, Santiago, Chile, August
  9-13, 2015, pp. 999--1002 (2015)

\bibitem{DBLP:journals/tip/WangLWZZH15}
Wang, Y., Lin, X., Wu, L., Zhang, W., Zhang, Q., Huang, X.: Robust subspace
  clustering for multi-view data by exploiting correlation consensus.
\newblock {IEEE} Trans. Image Processing \textbf{24}(11), 3939--3949 (2015)

\bibitem{DBLP:conf/cikm/WangLZ13}
Wang, Y., Lin, X., Zhang, Q.: Towards metric fusion on multi-view data: a
  cross-view based graph random walk approach.
\newblock In: 22nd {ACM} International Conference on Information and Knowledge
  Management, CIKM'13, San Francisco, CA, USA, October 27 - November 1, 2013,
  pp. 805--810 (2013)

\bibitem{DBLP:conf/pakdd/WangLZW14}
Wang, Y., Lin, X., Zhang, Q., Wu, L.: Shifting hypergraphs by probabilistic
  voting.
\newblock In: Advances in Knowledge Discovery and Data Mining - 18th
  Pacific-Asia Conference, {PAKDD} 2014, Tainan, Taiwan, May 13-16, 2014.
  Proceedings, Part {II}, pp. 234--246 (2014)

\bibitem{DBLP:journals/corr/abs-1708-02288}
Wang, Y., Wu, L.: Beyond low-rank representations: Orthogonal clustering basis
  reconstruction with optimized graph structure for multi-view spectral
  clustering.
\newblock Neural Networks \textbf{103}, 1--8 (2018)

\bibitem{NNLS2018}
Wang, Y., Wu, L., Lin, X., Gao, J.: Multiview spectral clustering via
  structured low-rank matrix factorization.
\newblock {IEEE} Trans. Neural Networks and Learning Systems  (2018)

\bibitem{DBLP:conf/ijcai/WangZWLFP16}
Wang, Y., Zhang, W., Wu, L., Lin, X., Fang, M., Pan, S.: Iterative views
  agreement: An iterative low-rank based structured optimization method to
  multi-view spectral clustering.
\newblock In: Proceedings of the Twenty-Fifth International Joint Conference on
  Artificial Intelligence, {IJCAI} 2016, New York, NY, USA, 9-15 July 2016, pp.
  2153--2159 (2016)

\bibitem{DBLP:journals/tnn/WangZWLZ17}
Wang, Y., Zhang, W., Wu, L., Lin, X., Zhao, X.: Unsupervised metric fusion over
  multiview data by graph random walk-based cross-view diffusion.
\newblock {IEEE} Trans. Neural Netw. Learning Syst. \textbf{28}(1), 57--70
  (2017)

\bibitem{DBLP:journals/ivc/WuW17}
Wu, L., Wang, Y.: Robust hashing for multi-view data: Jointly learning low-rank
  kernelized similarity consensus and hash functions.
\newblock Image Vision Comput. \textbf{57}, 58--66 (2017)

\bibitem{DBLP:journals/pr/WuWGL18}
Wu, L., Wang, Y., Gao, J., Li, X.: Deep adaptive feature embedding with local
  sample distributions for person re-identification.
\newblock Pattern Recognition \textbf{73}, 275--288 (2018)

\bibitem{DBLP:journals/cviu/WuWGHL18}
Wu, L., Wang, Y., Ge, Z., Hu, Q., Li, X.: Structured deep hashing with
  convolutional neural networks for fast person re-identification.
\newblock Computer Vision and Image Understanding \textbf{167}, 63--73 (2018)

\bibitem{TC2018}
Wu, L., Wang, Y., Li, X., Gao, J.: Deep attention-based spatially recursive
  networks for fine-grained visual recognition.
\newblock {IEEE} Trans. Cybernetics  (2018)

\bibitem{DBLP:journals/pr/WuWLG18}
Wu, L., Wang, Y., Li, X., Gao, J.: What-and-where to match: Deep spatially
  multiplicative integration networks for person re-identification.
\newblock Pattern Recognition \textbf{76}, 727--738 (2018)

\bibitem{LINYANG13}
Wu, L., Wang, Y., Shepherd, J.: Efficient image and tag co-ranking: a bregman
  divergence optimization method.
\newblock In: ACM Multimedia (2013)

\bibitem{DBLP:conf/sigir/YangC06}
Yang, H., Callan, J.P.: Near-duplicate detection by instance-level constrained
  clustering.
\newblock In: {SIGIR} 2006: Proceedings of the 29th Annual International {ACM}
  {SIGIR} Conference on Research and Development in Information Retrieval,
  Seattle, Washington, USA, August 6-11, 2006, pp. 421--428 (2006)

\bibitem{DBLP:conf/mm/ZhangC04}
Zhang, D., Chang, S.: Detecting image near-duplicate by stochastic attributed
  relational graph matching with learning.
\newblock In: Proceedings of the 12th {ACM} International Conference on
  Multimedia, New York, NY, USA, October 10-16, 2004, pp. 877--884 (2004)

\bibitem{DBLP:journals/pami/ZhangYWLT15}
Zhang, S., Yang, M., Wang, X., Lin, Y., Tian, Q.: Semantic-aware co-indexing
  for image retrieval.
\newblock {IEEE} Trans. Pattern Anal. Mach. Intell. \textbf{37}(12), 2573--2587
  (2015)

\bibitem{DBLP:conf/icip/ZhouLZ16}
Zhou, D., Li, X., Zhang, Y.: A novel cnn-based match kernel for image
  retrieval.
\newblock In: 2016 {IEEE} International Conference on Image Processing, {ICIP}
  2016, Phoenix, AZ, USA, September 25-28, 2016, pp. 2445--2449 (2016)

\bibitem{DBLP:conf/icimcs/ZhouLWLT12}
Zhou, W., Li, H., Wang, M., Lu, Y., Tian, Q.: Binary {SIFT:} towards efficient
  feature matching verification for image search.
\newblock In: The 4th International Conference on Internet Multimedia Computing
  and Service, {ICIMCS} '12, Wuhan, China, September 9-11, 2012, pp. 1--6
  (2012)

\bibitem{DBLP:journals/csur/ZobelM06}
Zobel, J., Moffat, A.: Inverted files for text search engines.
\newblock {ACM} Comput. Surv. \textbf{38}(2), 6 (2006)

\end{thebibliography}


\begin{thebibliography}{}
%
%
\bibitem{RefJ}
Author, Article title, Journal, Volume, page numbers (year)
\bibitem{RefB}
Author, Book title, page numbers. Publisher, place (year)
\end{thebibliography}

\end{document}